\documentclass[12pt]{article}
\usepackage[utf8]{inputenc}
\usepackage{hyperref}
\usepackage{fullpage}
\usepackage{amssymb}
\usepackage{mathtools}
\usepackage{amsthm, amsmath}
\usepackage{eucal}
\usepackage{dsfont}
\usepackage[T1]{fontenc}
\usepackage{lmodern}
\usepackage[stretch=10,shrink=10]{microtype}
\usepackage[capitalise]{cleveref}
\usepackage{color,soul}
\usepackage{tikz-cd}
\usepackage{subcaption}

\allowdisplaybreaks[1]

\theoremstyle{plain}
\newtheorem{theorem}{Theorem}[section]
\newtheorem{lemma}[theorem]{Lemma}

\theoremstyle{definition}

\newtheorem{claim}[theorem]{Claim}

\newcommand {\br} [1] {\ensuremath{ \left( #1 \right) }}

\newcommand {\bra} [1] {\ensuremath{ \left\langle #1 \right| }}
\newcommand {\ket} [1] {\ensuremath{ \left| #1 \right\rangle }}
\newcommand {\ketbratwo} [2] {\ensuremath{ \left| #1 \middle\rangle \middle\langle #2 \right| }}
\newcommand {\ketbra} [1] {\ketbratwo{#1}{#1}}

\newcommand {\id} {\ensuremath{\mathds{1}}}
\newcommand{\braket}[2]{\langle#1|#2\rangle}

\newcommand {\suppress}[1]{}

\newcommand {\defeq} {\ensuremath{ \stackrel{\mathrm{def}}{=} }}
\def\max{\mathrm{max}}
\def\min{\mathrm{min}}

\newcommand{\bigo}[1]{\mathcal{O}\left(#1\right)}
\newcommand{\omeg}[1]{\Omega\left(#1\right)}

\newcommand{\vi}{\bf{i}}

\newcommand{\step}{\mathrm{Step}}
\newcommand{\intm}{\lceil m \rceil}
\newcommand{\nbyt}{\mathsf{quo}}

\def\cL{\mathcal{L}}

\hypersetup{
	pdfstartview={FitH},
	pdfdisplaydoctitle={true},
	breaklinks={true},
	bookmarksopen={true},
	bookmarksnumbered={false},
}

\begin{document}

\title{\textbf{Improved local spectral gap thresholds for lattices of finite dimension}}

\author{Anurag Anshu}

\maketitle

\begin{abstract}
Knabe's theorem lower bounds the spectral gap of a one dimensional frustration-free local hamiltonian in terms of the local spectral gaps of finite regions. It also provides a local spectral gap threshold for hamiltonians that are gapless in the thermodynamic limit, showing that the local spectral gap much scale inverse linearly with the length of the region for such systems. Recent works have further improved upon this threshold, tightening it in the one dimensional case and extending it to higher dimensions. Here, we show a local spectral gap threshold for frustration-free hamiltonians on a finite dimensional lattice, that is optimal up to a constant factor that depends on the dimension of the lattice. Our proof is based on the detectability lemma framework and uses the notion of coarse-grained hamiltonian (introduced in [Phys. Rev. B 93, 205142]) as a link connecting the (global) spectral gap and the local spectral gap.
\end{abstract}

\section{Introduction}

A central problem in condensed matter physics is to understand the properties of the ground states of spin systems. While finding a complete description of the ground states can be a daunting task, many important ground state properties (notably the area laws \cite{Hastings07, AradLV12, AradKLV13} and the decay of correlation \cite{Hastings04, HastingsK06, NachtergaeleS06}) are intertwined with the spectral gap of the associated hamiltonian. Thus, understanding the spectral gap of a local hamiltonian takes a central stage in the mathematical physics of spin systems.

Recent results on the undecidability of the spectral gap \cite{CubittPW15, BauschCLP19} show that there is no general scheme for computing the spectral gap of an arbitrary local hamiltonian. But, for a large and important family known as the frustration-free hamiltonians (to be defined shortly), there are two powerful methods that provide criteria for system size independent lower bounds on the spectral gap. First is the martingale method due to Nachtergaele \cite{Nachtergaele96}, that guarantees a large spectral gap whenever certain product of the local ground space projectors is close to the global ground space projector. The second method, introduced by Knabe \cite{Knabe88}, bounds the spectral gap whenever the `local' spectral gap in a finite region is large enough.  These tools have found several applications in recent years, such as in the classification of gapped phases for qubits \cite{BravyiG15}, gap of generic translationally invariant hamiltonians \cite{Lemm19a}, properties of random quantum circuits \cite{BrandaoHH16} etc.

A local hamiltonian $H=\sum_\alpha P_\alpha$, where each $P_\alpha \succeq 0$ is a local interaction that acts on a small number of spins, is said to be frustration-free if its ground space $G$ satisfies $P_\alpha G=0$ for all $\alpha$. Without loss of generality, one can assume that each $P_{\alpha}$ is a projector (that is, $P^2_{\alpha}=P_\alpha$), with a constant multiplicative change in the spectral gap. Knabe's method, which is the central focus of present work, applies to translationally invariant frustration-free hamiltonians on a periodic chain of spins. It states the following.

\vspace{0.1in}

\noindent {\bf Knabe's theorem:} Let $H= \sum_{i=1}^n P_{i,i+1}$ be a translationally invariant nearest-neighbour hamiltonian on a periodic chain of $n$ spins, with spectral gap $\gamma$. Let $h_{k,t} = \sum_{i=k+1}^{k+t-1} P_{i,i+1}$ be the hamiltonian restricted to the spins $\{k+1,k+2,\ldots k+t\}$. Let $\gamma(t)$ be the spectral gap of $h_{k,t}$, which does not depend on $k$ due to translation invariance. Then $\gamma + \frac{1}{t-2} \geq \frac{t-1}{t-2}\gamma(t)$.

\vspace{0.1in}

We provide a sketch of the proof, to help compare with our techniques. The spectral gap $\gamma$ of $H$ is the largest number that satisfies $H^2 \succeq \gamma H$. In order to lower bound $\gamma$, we expand $H^2= \sum_{i, i'}P_{i,i+1} P_{i',i'+1}$ and use the fact that $P_{i,i+1}$ are projectors to simplify $H^2= H + \sum_{i\neq i'}P_{i,i+1} P_{i',i'+1}$. If all the terms $P_{i,i+1} P_{i',i'+1}$ were positive semi-definite, we would obtain $H^2 \succeq H$, leading to $\gamma \geq 1$. But this is not the case, as overlapping local terms $P_{i,i+1}$ and $P_{i+1,i+2}$ need not commute. To handle such terms, Knabe \cite{Knabe88} invokes the hamiltonians $h_{k,t}$ and makes use of the operator inequality $h_{k,t}^2 \succeq \gamma(t) h_{k,t}$. This helps in lower bounding sums of the form $\sum_{i,i'}P_{i,i+1} P_{i',i'+1}$ in terms of $\gamma(t)$.
 
An important consequence of Knabe's theorem is that if $H$ is gapless in the thermodynamic limit (that is, $\gamma\rightarrow 0$ as $n\rightarrow \infty$), then the local gap $\gamma(t)$ must be less than $\frac{1}{t-1}$. This is often termed as the `local gap threshold'. The additive term that captures the local gap threshold has been improved in the recent work \cite{GossetM16}, which shows the inequality $\gamma + \frac{5}{t^2-4} \geq \frac{5}{6}\gamma(t)$. This inequality is tight up to constants, as witnessed by the Heisenberg ferromagnet (see \cite[Section 2]{GossetM16} for details). The authors also consider the problem on a two dimensional periodic square lattice $\cL$, with a nearest-neighbour translationally invariant hamiltonian $H=\sum_{e}P_e$. Here the index $e$ runs over the edges of the lattice. In the same spirit as above, they obtain the inequality $\gamma + \frac{6}{t^2} \geq \gamma(t)$, where $\gamma$ is the spectral gap of $H$ and $\gamma(t)$ is the spectral gap of the hamiltonian $h_S$ restricted over a square region $S$ of side length $t$. 

Subsequent works have made further progress in this direction. The results of \cite{GossetM16} have been extended to a two dimensional lattice with open boundary conditions in \cite{LemmM18}, with the additive term scaling as $t^{-3/2}$. The work \cite{KastoryanoL18} shows that for a gapless hamiltonian on a lattice $\cL$ of finite dimension, the local gap threshold scales as $\bigo{\frac{\log^2(t)}{t}}$. Remarkably, it builds upon the martingale method \cite{Nachtergaele96} and the detectability lemma \cite{AharonovALV08}, rather than the techniques in \cite{Knabe88, GossetM16, LemmM18} sketched earlier. More recently, \cite{Lemm19} improves this to an upper bound of $\frac{3}{t}$ (on a finite dimensional lattice) for the hyper-cubic regions of side length $t$. 

\section{Our results}

Consider a $D$ dimensional regular lattice $\cL$ with unit cells as hypercubes and spins situated on the vertices. Let $H$ be a local hamiltonian defined as $H=\sum_e P_e$, where $e$ runs over the unit cells of $\cL$ and $P_e$ is supported only on the $2^D$ vertices of the corresponding unit cell. This particular set-up is chosen for convenience and our results can be generalized to other lattices as long as the interactions $P_e$ are local. As before, let $\gamma$ be the spectral gap of $H$. For a tuple of integers $(t_1, \ldots t_D)$, we let $\gamma(t_1, t_2, \ldots t_D)$ denote the minimum spectral gap over all hamiltonians $h_S$ restricted to hyper-rectangles $S$ of size $t_1\times t_2\times \ldots t_D$ (where $t_i$ is the side length along the $i$-th axis). We show that
\begin{equation}
\label{eq:infmaintheo}
\gamma(t_1, t_2, \ldots t_D) = \bigo{\gamma + \frac{1}{\min_{q} t^2_q}},
\end{equation}
 where the notation $\bigo{.}$ hides the factors that depend on $D$ (see the formal statement in Theorem \ref{theo:rectknabe}). Note that we do not require $H$ to be translationally invariant. The statement applies to both the open and periodic boundary conditions. For hyper-cubic regions with $t_1=t_2=\ldots t_D=t$, the additive term scales as $\bigo{\frac{1}{t^2}}$, improving upon prior works for $t$ larger than a constant that depends on $D$. 

As discussed towards the end of Section \ref{sec:Ddimgap}, the additive term of $\frac{1}{\min_q t^2_q}$ cannot be improved even in the translationally-invariant case (as witnessed by many parallel copies of a chain of Heisenberg ferromagnet), except potentially for the constant that depends on $D$. Further, Equation \ref{eq:infmaintheo} would be false if $\gamma(t_1, t_2, \ldots t_D)$ were defined as an average (instead of a minimum) over hyper-rectangles of size $t_1\times t_2\times \ldots t_D$.

\subsection{Proof outline}

It suffices to consider the one dimensional case to discuss the proof technique. We will explain later that the higher dimensional case is a simple recursive application of this one dimensional argument. Consider the one dimensional nearest-neighbour hamiltonian $H=\sum_i P_{i,i+1}$ on an open chain of spins, with spectral gap $\gamma$ and ground space $G$. Let $\gamma(t)$ be the minimum spectral gap over all hamiltonians $\sum_{i=k+1}^{k+t-1} P_{i,i+1}$, where $k\in \{0,1,\ldots n-t\}$ . Central to our argument is the coarse-grained hamiltonian $\bar{H}(t) = \sum_S Q_S$ from \cite{AAV16}, which has the same ground space $G$. Here, $S$ are some sets of $t$ consecutive spins (see Figure \ref{fig:tseg}) and $Q_S$ project onto the non-zero eigenstates of $\sum_{i,i+1\in S}P_{i,i+1}$ . Let $\gamma(\bar{H}(t))$ be the spectral gap of $\bar{H}(t)$. The coarse-grained hamiltonian provides a link between $\gamma$ and $\gamma(t)$, as made precise in the following observation \cite{Gossetper19}:
\begin{equation}
\label{eq:gossetineq}
\gamma(\bar{H}(t)) \leq \frac{2\gamma}{\gamma(t)}.
\end{equation}
Its formal proof (in slight generality incorporating the higher dimensional lattices) will be given in Subsection \ref{subsec:prooftheoknabe}. It was shown in \cite{AAV16} that for $t= \omeg{\frac{1}{\sqrt{\gamma}}}$,  $\gamma(\bar{H}(t))=\Omega(1)$. This immediately says that $\gamma(t)=\bigo{\gamma}$ for this choice of $t$. An extension of this result to all $t$ relies on an estimate of the `shrinking ability' of low degree Chebyshev polynomials, which is shown in Claim \ref{clm:chebylowdeg} (see also \cite[Theorem 42]{EldarH17} for a similar estimate). It shows that $\gamma(\bar{H}(t))=\omeg{\frac{t^2\gamma}{1+t^2\gamma}}$, using the converse to the detectability lemma (Lemma \ref{convdetectorig}). Plugging in Equation \ref{eq:gossetineq}, we find that $\gamma(t) = \bigo{\gamma + \frac{1}{t^2}}$. Note that we did not require translation invariance and argument can easily be modified for the periodic chain, by considering a similar coarse-grained hamiltonian.

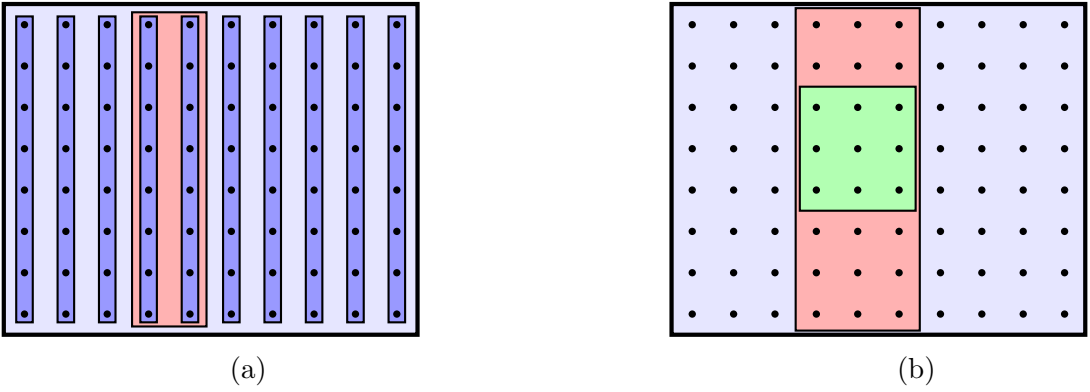
\begin{figure}
\centering
\begin{subfigure}[b]{0.4\textwidth}
\begin{tikzpicture}[xscale=0.55,yscale=0.55]

\draw [fill=blue!10!white, ultra thick] (0.5,0.5) rectangle (10.5, 8.5);
\draw [fill=red!30!white, thick] (3.6,0.7) rectangle (5.4, 8.3);
\foreach \i in {1,...,10}
{
\draw [fill=blue!40!white, thick] (\i-0.2, 0.8) rectangle (\i+0.2, 8.2);
\foreach \j in {1,...,8}
   \draw (\i, \j) node[circle, fill=black, scale=0.25]{};
}

\end{tikzpicture}
\caption{}
  \end{subfigure}
\hspace{2cm}
\begin{subfigure}[b]{0.4\textwidth}
\begin{tikzpicture}[xscale=0.55,yscale=0.55]

\draw [fill=blue!10!white, ultra thick] (0.5,0.5) rectangle (10.5, 8.5);

\draw [fill=red!30!white, thick] (3.5,0.6) rectangle (6.5, 8.4);

\draw [fill=green!30!white, thick] (3.6,3.5) rectangle (6.4, 6.5);

\foreach \i in {1,...,10}
{
\foreach \j in {1,...,8}
   \draw (\i, \j) node[circle, fill=black, scale=0.25]{};
}

\end{tikzpicture}
  \caption{}
\end{subfigure}
\caption{\small (a) We can view the hamiltonian on two dimensional lattice as a hamiltonian on one dimensional chain of column of spins (dark blue rectangles). The interaction $H_4$ between columns $4$ and $5$ is shown as the red rectangle, which decomposes as $H_4= \sum_{j=1}^{n_2-1}P_{4,j}$. (b) Our strategy is to lower bound the spectral gap of $H$ with the spectral gap of hamiltonian $h_S$ supported on the red region $S\times \{1,2,\ldots n_2\}$. The spectral gap of $h_S$ can in turn be lower bounded by the spectral gap of the hamiltonian $h_{S,S'}$ supported on the green region $S\times S'$. \label{fig:colqudit}}
\end{figure}

To explain the argument for higher dimensional lattices, consider a hamiltonian $H=\sum_{i=1}^{n_1-1}\sum_{j=1}^{n_2-1}P_{i,j}$ on a two dimensional square lattice $\{1,2,\ldots n_1\}\times \{1,2,\ldots n_2\}$, where $P_{i,j}$ is supported on the spins $\{(i,j), (i+1,j), (i,j+1), (i+1,j+1)\}$. Following \cite{AAG19}, we can view this hamiltonian as a one dimensional hamiltonian $H=\sum_{i=1}^{n_1-1} H_i$, where $H_i=\sum_{j=1}^{n_2-1}P_{i,j}$ is the `column' hamiltonian acting on two columns of spins, that is, $\{\forall j: (i,j)\}$ and $\{\forall j: (i+1,j)\}$ (c.f. Figure \ref{fig:colqudit}(a)). Such a one dimensional view is not helpful for the technique used in \cite{Knabe88, GossetM16, LemmM18, Lemm19}, as the column hamiltonians $H_i$ are not projectors (recall the sketch of the proof given in the introduction, which crucially uses the fact that $P_{i,i+1}$ are projectors). But our method can be applied to a sum of column hamiltonians in the same manner as the one dimensional case.  We relate the spectral gap of $H$ to the spectral gap of $h_S\defeq \sum_{i \in S}H_i$, where $S$ is some continuous subset of $\{1,2,\ldots n_1\}$ of size $t_1$, up to the additive factor of $\bigo{\frac{1}{t_1^2}}$. Now, $h_S$ is a local hamiltonian on $t\times n_2$ spins (red region in Figure \ref{fig:colqudit}(b)) and can also be viewed as a sum $\sum_{j=1}^{n_2-1} H'_j$ of `row' hamiltonians $H'_j\defeq \sum_{i\in S} P_{i,j}$ acting on rows of spins $\{\forall i\in S: (i,j)\}$ and $\{\forall i\in S: (i,j+1)\}$. Thus, we can apply the same argument to $h_S$, relating its spectral gap to the spectral gap of some local hamiltonian $h_{S,S'}\defeq \sum_{i\in S, j\in S'} P_{i,j}$ (green region in Figure \ref{fig:colqudit}(b)). Here $S'$ is a set of size $t_2$, implying that $h_{S,S'}$ is supported on a square region of size $t_1\times t_2$. The overall additive factor is 
$\bigo{\frac{1}{t_1^2}+\frac{1}{t_1^2}}=\bigo{\frac{1}{\min_{q\in \{1,2\}} t_q^2}}$. Same recursive argument applies to higher dimensions.

\vspace{0.1in}

\noindent {\bf Comparison to prior work:} As already mentioned, our tools significantly differ from those employed in \cite{Knabe88, GossetM16, LemmM18, Lemm19}. Similar to us, the work \cite{KastoryanoL18} employs the detectability lemma and its converse to obtain the local gap threshold. But it does not use the coarse-grained hamiltonians, and builds upon the martingale method. We remark that it may be possible to improve their local gap threshold from $\bigo{\frac{\log^2(t)}{t}}$ to $\bigo{\frac{\text{poly}(\log(t))}{t^2}}$. This is because the statement in \cite[Theorem 11]{KastoryanoL18} can be improved using the ideas presented in \cite{GossetH15}. Such an improvement would still be slightly weaker than our bound in Equation \ref{eq:infmaintheo}, which does not contain the $\text{poly}(\log(t))$ factor. 

\vspace{0.1in}

\noindent {\bf Organization of the technical part:} The technical details appear in the Appendix. Section \ref{sec:chebylowdeg} derives the shrinking factor of low degree Chebyshev polynomials. Section \ref{sec:setupres} defines a general model of local hamiltonian on a chain of spins, which encompasses the column and row hamiltonians discussed earlier. Our main result is Theorem \ref{theo:knabe} which applies to this model of local hamiltonian (for both periodic and open boundary conditions). We discuss the tools of detectability lemma and the coarse-grained hamiltonian in the same section.  Proof of the main result (Theorem \ref{theo:knabe}) appears in Subsection \ref{subsec:prooftheoknabe}. Section \ref{sec:Ddimgap} recursively uses Theorem \ref{theo:knabe} to obtain the result in Equation \ref{eq:infmaintheo}. Sections \ref{append:convdetect} and \ref{append:lightcone} re-derive some known results for completeness. 

\subsection*{Conclusion}

In this work, we have derived a relation between the (global) spectral gap and the local spectral gap of frustration-free local hamiltonians on a lattice, along the lines of Knabe \cite{Knabe88}. The relation is optimal up to factors that depend on the dimension of the lattice. It may be potentially improved if the hamiltonian has further symmetry. For concreteness, consider a local hamiltonian $H=\sum_e P_e$, where $e$ runs over the edges of the lattice and $P_e$ is the same interaction across every edge (in other words, the hamiltonian is isotropic and translationally invariant). In this case, we conjecture that the additive term in Equation \ref{eq:infmaintheo} can be improved to $\frac{1}{\sum_q t^2_q}$, which is the inverse-squared diameter of the hyper-rectangles.  

Our proof is based on the technique of coarse-grained hamiltonian introduced in \cite{AAV16} and shows how the detectability lemma \cite{AharonovALV08} can be used to capture yet another feature of the frustration-free systems. It would be interesting to apply our method to bound the spectral gaps of specific models of frustration-free hamiltonians on a two dimensional lattice (see \cite{HoussamLLNY18, LemmSY19} for such recent applications of prior techniques). It would also be interesting to find implications of our results to the existence of chiral edge modes in three or more dimensions (c.f \cite{LemmM18}).

\subsection*{Acknowledgement}

I am grateful to David Gosset for discussions related to this work, and for sharing his observation in Equation \ref{eq:gossetineq}. I thank Dorit Aharonov, Itai Arad, Fernando Br{\~a}ndao, Angelo Lucia, Marius Lemm and Jamie Sikora for helpful discussions. This work is supported by the Canadian Institute for Advanced Research, through funding provided to the Institute for Quantum Computing by the Government of Canada and the Province of Ontario. Perimeter Institute is also supported in part by the Government of Canada and the Province of Ontario.

\bibliographystyle{alpha}
\bibliography{references1}

\appendix

\section{Low degree behaviour of Chebyshev polynomials}
\label{sec:chebylowdeg}

\noindent Chebyshev polynomial of degree $m$ is defined as 
$$T_m(x) = \begin{cases}  \cos(m\arccos(x)) , \text{ if } |x| < 1\\ \cosh(m\cosh^{-1}(x)), \text{ if } |x|\geq 1 \end{cases}.$$
It has found applications in area laws \cite{AradLV12, AradKLV13}, sub-volume law \cite{AAG19} and the decay of correlation \cite{GossetH15}. We have the following claim (see also \cite[Theorem 42]{EldarH17}). 
\begin{claim}
\label{clm:chebylowdeg}
Fix $\nu\in \br{0,\frac{1}{4}}$ and a real number $m> 0$. Consider the polynomial
$$\step_{m, \nu}(x) = \frac{T_{\intm}\br{-1+\frac{2x}{1-\nu}}}{T_{\intm}\br{\frac{1+\nu}{1-\nu}}}.$$ It holds that
$\step_{m,\nu}(1)=1$ and 
$$|\step_{m, \nu}(x)|\leq \frac{1}{1+ \frac{m^2\nu}{2(1-\nu)}},$$ for $x\in (0,1-\nu)$.
\end{claim}
\begin{proof}
The relation $\step_{m,\nu}(1)=1$ trivially holds. Since $T_{\intm}\br{-1+\frac{2x}{1-\nu}}\in \{-1,1\}$ for $x\in (0, 1-\nu)$, we have $|\step_{m, \nu}(x)| \leq \frac{1}{T_{\intm}\br{\frac{1+\nu}{1-\nu}}}$ for $x\in (0, 1-\nu)$. We wish to upper bound $\frac{1}{T_{\intm}\br{\frac{1+\nu}{1-\nu}}}$. Let $w$ be such that $\cosh(w)= \frac{1+\nu}{1-\nu}$. Then 
\begin{equation}
\label{chebylowb}
T_{\intm}\br{\frac{1+\nu}{1-\nu}} = \cosh(\intm w) \geq 1+ \frac{\intm^2w^2}{2} \geq 1+ \frac{m^2w^2}{2}.
\end{equation}
 Now,
$$\frac{1+\nu}{1-\nu} = \cosh(w) = \frac{e^w+e^{-w}}{2} \implies \frac{2\nu}{1-\nu} = \frac{e^w+e^{-w}-2}{2} = \frac{(e^{\frac{w}{2}} - e^{-\frac{w}{2}})^2}{2}.$$ This implies
$$e^{\frac{w}{2}} - e^{-\frac{w}{2}} = 2\sqrt{\frac{\nu}{1-\nu}}.$$ Solving the quadratic equation for $e^{\frac{w}{2}}$, we find 
$$e^{\frac{w}{2}} = \sqrt{1+ \frac{\nu}{1-\nu}} + \sqrt{\frac{\nu}{1-\nu}} \geq 1+ \sqrt{\frac{\nu}{1-\nu}}.$$ Thus, $$w \geq 2\log\br{1+ \sqrt{\frac{\nu}{1-\nu}}} \geq \sqrt{\frac{\nu}{1-\nu}},$$ for $\nu \leq \frac{1}{4}$. Equation \ref{chebylowb} now implies
$$T_{\intm}\br{\frac{1+\nu}{1-\nu}} \geq 1+ \frac{m^2w^2}{2} \geq 1+ \frac{m^2\nu}{2(1-\nu)},$$
which leads to 
$$|\step_{m, \nu}(x)| \leq \frac{1}{T_{\intm}\br{\frac{1+\nu}{1-\nu}}} \leq \frac{1}{1+ \frac{m^2\nu}{2(1-\nu)}},$$ for $x\in (0,1-\nu)$.
\end{proof}

\section{Formal set-up and the main result}
\label{sec:setupres}

Here, we introduce notations to analyze both the open chain and closed chain of qudits \footnote{We have shifted to the terminology `qudits' instead of `spins', which is more standard in quantum information.} . Let $[a:b]$ denote the set $\{a,a+1, \ldots b\}$. Consider a one dimensional closed chain of $n$ qudits, indexed by integers $\{1,2,\ldots n\}$, of potentially varying dimensions. The indices of the qudits are taken in a manner that the $n+k$-th index is the same as the $k$-th index.  Introduce a nearest-neighbour local hamiltonian
\begin{equation}
H=\sum_{i=1}^n H_i
\label{eq:twoloc}
\end{equation}
where $H_i$ is a Hermitian operator which acts nontrivially only on qudits $i,i+1$. Further assume that $H_i$ admits a decomposition
\begin{equation}
H_i = \sum_j P_{ij},
\label{eq:Hidecompose}
\end{equation}
where $P_{ij}$ are projectors that act non trivially only on qudits $i,i+1$. We have the following assumptions on the set of projectors $\{P_{ij}\}_{i,j}$:
\begin{itemize}
\item Each $P_{ij}$ does not commute with at most $g$ other terms from the set $\{P_{ij}\}_{i,j}$.
\item The projectors can be divided into $L$ layers $T_1, T_2, \ldots T_L$, where the terms within each layer mutually commute.
\end{itemize}
We further assume the $H$ is frustration-free, which means that the ground energy is zero. Let $G$ be the ground space of $H$. Note that frustration-freeness implies that $P_{ij}G=0$. We shall write $G_{\perp}$ for the subspace of states orthogonal to $G$. 

If we are interested in an open chain of qudits, then we simply assume that $H_n=0$ (note that we are not considering the translationally invariant case). This will lead to some minor changes that we will highlight as the arguments proceed.

For a contiguous subset $S$ of the chain, let $h_S\defeq \sum_{i: i, i+1\in S}H_i$ be the local hamiltonian made out of terms in Eq.~\ref{eq:twoloc} that are entirely supported in $S$. Define $\gamma(S)$ to be the smallest nonzero eigenvalue of $h_S$  and let $\gamma\defeq \gamma\br{[1:n]}$ be the spectral gap of $H$. Let $$\gamma(t)\defeq \min_{a} \gamma\br{[a:a+t-1]}$$ denote the minimum spectral gap over all continuous segments of length $t$. Observe that the set of continuous segments are different for the open chain and the closed chain. Thus the minimization over $a$ in above expression requires the additional condition that $a\in [1:n-t+1]$ for the open chain. Our main theorem is as follows, which upper bounds $\gamma(t)$  in terms of the spectral gap $\gamma$. The statement remains the same for both the open chain and the closed chain.
\begin{theorem}
\label{theo:knabe}
Suppose $\gamma\leq \frac{g^2}{4}$. For every integer $8L^2<t<n/5$, it holds that
$$\gamma(t)\leq \frac{10^3L^2g^2}{t^2} + 6\gamma.$$
\end{theorem}

Note that we have not tried to optimize the parameters appearing in the above expression. For specific applications, it may be possible to obtain stronger bounds. The rest of the section is devoted to the proof of Theorem \ref{theo:knabe}. 

\subsection{Detectability lemma}

Detectability lemma \cite{AharonovALV08} is an important tool for the study of frustration-free systems. It's central object is the detectability lemma operator, defined as a product of projectors $\id-P_{ij}$ taken layer by layer. More precisely, define 
$$DL(H)\defeq \prod_{\alpha\in [1:L]}\prod_{i,j\in T_{\alpha}}(\id-P_{ij}).$$ The following lemma holds, the statement of which is taken from \cite[Corollary 3]{AAV16}.

\begin{lemma}[Detectability lemma, \cite{AharonovALV08}]
\label{detectlem}
For any quantum state $\psi\in G_{\perp}$, we have 
$$\|DL(H)\ket{\psi}\|^2 \leq \frac{1}{1+\gamma/ g^2}.$$
\end{lemma}

A converse result stated in \cite[Lemma 4]{AAV16} is a corollary of \cite{Gao2015}. 
\begin{lemma}[Converse to detectability lemma, \cite{AAV16, Gao2015}]
\label{convdetectorig}
For any quantum state $\psi$,
$$\|DL(H)\ket{\psi}\|^2 \geq 1- 4\bra{\psi}H\ket{\psi}.$$
\end{lemma}
Here we provide a short proof (with minor improvement) in the special case of $L=2$. Proof is deferred to Appendix \ref{append:convdetect}.
\begin{lemma}
\label{convdetect}
Suppose $L=2$. It holds that
$$\|DL(H)\ket{\psi}\|^2 \geq 1- 3\bra{\psi}H\ket{\psi}.$$
\end{lemma}
\subsection{Coarse-grained Hamiltonian}

Another tool that we will use is the notion of coarse-grained Hamiltonian \cite{AradLV12, AAV16}. Let $Q_S$ be the projector orthogonal to the ground space of $h_S$. By convention, we set $Q_{\phi}=0$ for the empty set $\phi$. Fix a coarse-graining parameter $8L^2< t < n/5$ and let $\nbyt = \lfloor \frac{n}{t} \rfloor$ and $r = n- t\cdot\nbyt$ respectively be the quotient and remainder when $n$ is divided by $t$. Identify sets $S_1, S_2, \ldots S_{\nbyt}$ using the following rules. 
\begin{itemize}
\item $S_k \defeq [s_k: s'_k]$ with $1\leq s_1 < s'_1 < s_2 < s'_2 \ldots < s_{\nbyt} < s'_{\nbyt} \leq n$. Further, $|S_k|=t$.
\item Let $r_k = s_{k+1}-s'_k-1$ be the number of qudits sandwiched between $S_k, S_{k+1}$, for $k\in [1:\nbyt-1]$. Let $r_{\nbyt}= n- s'_{\nbyt} + s_1-1$ be the number of qudits sandwiched between $S_{\nbyt}$ and $S_1$. Observe that $\sum_{k=1}^{\nbyt}r_k=r$. We require that $r_k$ are not too large. That is, $r_k\leq \lceil r/\nbyt \rceil$ for all $k$. Since $\frac{n}{t}>5$, this implies that 
\begin{equation}
\label{eq:rkupb}
r_k\leq \lceil r/5 \rceil \leq t/4.
\end{equation}
\item For open chain, with $H_n=0$, we require $s_1=1$ and $s'_{\nbyt}=n$. 
\end{itemize}
Next, choose another collection of $\nbyt$ continuous sets of size $t$ each, which are placed `halfway' between adjacent $S$'s. More precisely, the sets $T_1, T_2, \ldots T_{\nbyt}$ have the following properties.
\begin{itemize}
\item For $k<\nbyt$, $T_k = [s'_k - \lfloor \frac{t-r_k}{2} \rfloor+1 : s_{k+1} + \lceil \frac{t-r_k}{2}\rceil-1]$.
\item For the open chain (with $H_n=0$), let $T_{\nbyt}=\phi$. For the closed chain, let $$T_{\nbyt} = [s'_{\nbyt} - \lfloor \frac{t-r_{\nbyt}}{2} \rfloor+1 : s_1 + \lceil \frac{t-r_{\nbyt}}{2}\rceil-1].$$ 
\end{itemize}

Two examples of these sets are depicted in Figures \ref{fig:tseg} and \ref{fig:openseg}. Observe that the set $T_k$ has an overlap of at least $\lfloor \frac{t-r_k}{2} \rfloor$ with the sets $S_k$ and $S_{k+1}$. Using $t\geq 8L^2\geq 8$ and Equation \ref{eq:rkupb}, this can be lower bounded by 
\begin{equation}
\label{eq:overlaplb}
\lfloor\frac{t-r_k}{2} \rfloor\geq \lfloor \frac{t-t/4}{2} \rfloor = \lfloor \frac{3t}{8} \rfloor \geq \frac{3t}{8}-1 = \frac{t}{4} + \frac{t}{8}-1\geq \lfloor\frac{t}{4}\rfloor.
\end{equation}

\begin{figure}
\centering
\begin{tikzpicture}[xscale=0.4,yscale=0.6]
\foreach \k in {0,...,2}
{
\draw[draw=black, fill=green!30!white] (0.7+6*\k,-1.5) rectangle (5.3+6*\k,-0.5);

\draw[draw=black, fill=red!30!white] (3.7+6*\k,1.5) rectangle (8.3+6*\k,0.5);
}

\foreach \k in {3,...,5}
{
\draw[draw=black, fill=green!30!white] (2.7+5*\k,-1.5) rectangle (7.3+5*\k,-0.5);

\draw[draw=black, fill=red!30!white] (5.7+5*\k,1.5) rectangle (10.3+5*\k,0.5);
}

\draw[draw=black, fill=green!30!white] (2.7+5*6,-1.5) rectangle (7.3+5*6,-0.5);

\draw[draw=black, fill=blue!30!white] (5.7+5*6,1.5) rectangle (10.3+5*6-3,0.5);
\draw[draw=black, fill=blue!30!white] (0.7,1.5) rectangle (3.3,0.5);

\foreach \i in {1,...,37}
{
   \draw (\i, 0) node[circle, fill=black, scale=0.5]{};
}

\end{tikzpicture}
  \caption{\small Dividing the chain into contiguous segments of length $t$: Here, we assume $n=37$ and $t=5$. The remainder when $n$ is divided by $t$ is $2$. We set $r_1=r_2=1$ and $r_k=0$ for $k>2$. The green rectangles represent the sets $S_j$. The red and the blue rectangles represent the sets $T_j$. The blue rectangles are to be viewed as a single contiguous region on the closed chain when $H_n\neq 0$ and are assumed to not exist on the open chain when $H_n=0$. The first three rectangles, of both green and red kind, are separated by one qudit.  \label{fig:tseg}}
\end{figure}
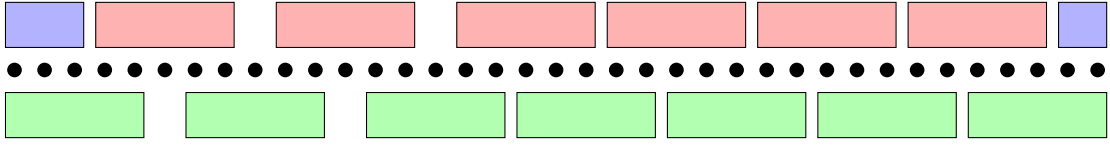

\begin{figure}
\centering
\begin{tikzpicture}[xscale=0.4,yscale=0.6]
\draw[draw=black, fill=green!30!white] (0.7,-1.5) rectangle (18.3,-0.5);
\draw[draw=black, fill=green!30!white] (20.7,-1.5) rectangle (38.3,-0.5);

\draw[draw=black, fill=red!30!white] (9.7,1.5) rectangle (27.3,0.5);

\foreach \i in {1,...,38}
{
   \draw (\i, 0) node[circle, fill=black, scale=0.5]{};
}

\end{tikzpicture}
  \caption{\small Assume $n=38$, $t=18$ and $H_n=0$ (open chain). In this case, $r=2$. There is exactly one set $T_1$ and two sets $S_1, S_2$.   \label{fig:openseg}}
\end{figure}
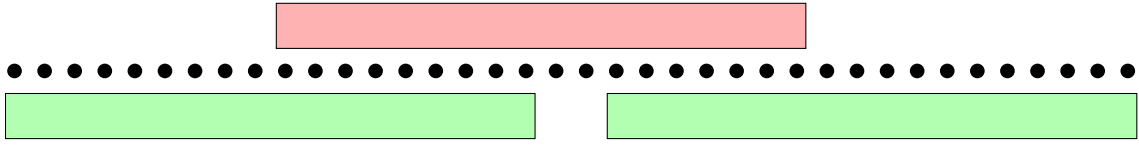

Following \cite{AAV16}, we define the coarse-grained hamiltonian 
$$\bar{H}(t)\defeq \sum_{k}\br{ Q_{S_k} + Q_{T_k}}$$
and the corresponding detectability operator
$$DL(t)\defeq \br{\prod_{k}\br{\id - Q_{S_k}}}\br{\prod_{k}\br{\id - Q_{T_k}}}.$$ 
Observe that the ground space of $\bar{H}(t)$ coincides with $G$. Let the spectral gap of $\bar{H}(t)$  be $\gamma(\bar{H}(t))$. The following lemma was shown in \cite{AAV16}.  We provide its proof in Appendix \ref{append:lightcone}, for completeness.
\begin{lemma}
\label{lightconelem}
It holds that
$$1- 3\gamma(\bar{H}(t))\leq\max_{\psi\in G_{\perp}}\|DL(t)\ket{\psi}\|^2 \leq \max_{x\in (0, 1- \frac{\gamma}{g^2+\gamma})}\step_{\frac{t}{8L}, \frac{\gamma}{g^2+\gamma}}(x).$$
\end{lemma}

Now, we proceed to the proof of our main theorem.

\subsection{Proof of Theorem \ref{theo:knabe}}
\label{subsec:prooftheoknabe}

We start with the inequality for all $1\leq k \leq \nbyt$
$$Q_{S_k} + Q_{T_k}\preceq \frac{1}{\gamma(S_k)}h_{S_k} + \frac{1}{\gamma(T_k)}h_{T_k}\preceq \frac{1}{\gamma(t)}\br{h_{S_k}+ h_{T_k}}.$$
Note that the above inequality also holds in the case of open chain, as $T_{\nbyt}=\phi$ implies $Q_{T_{\nbyt}}=0$ and $h_{T_{\nbyt}}=0$. Summing over $k$ and using the definition of $\bar{H}(t)$, this implies that
$$\bar{H}(t) \preceq \frac{1}{\gamma(t)}\sum_k \br{h_{S_k} + h_{T_k}} = \frac{1}{\gamma(t)}\sum_k\br{\sum_{i: \text{Supp}(H_i)\in S_k} H_i + \sum_{i: \text{Supp}(H_i)\in T_k} H_i } \preceq \frac{2}{\gamma(t)}H.$$ Here, the last inequality holds since each $H_i$ is supported within at most one $S_k$ and at most one $T_k$. As a result, we have the following inequality (stated earlier in Equation \ref{eq:gossetineq}):
\begin{equation}
\label{coarsegrainedgaps}
\gamma(\bar{H}(t)) = \min_{\psi\in G_{\perp}}\bra{\psi}\bar{H}(t)\ket{\psi} \leq \frac{2}{\gamma(t)} \min_{\psi\in G_{\perp}}\bra{\psi}H\ket{\psi} = \frac{2\gamma}{\gamma(t)}.
\end{equation}
Lemma \ref{lightconelem} ensures that 
$$\gamma(\bar{H}(t)) \geq \frac{1}{3}\br{1-\max_{x\in (0, 1- \frac{\gamma}{g^2+\gamma})}\step_{\frac{t}{8L}, \frac{\gamma}{g^2+\gamma}}(x)}.$$ Now we use Claim \ref{clm:chebylowdeg}, setting $m=\frac{t}{8L}$ and $\nu=\frac{\gamma}{g^2+\gamma} \leq \frac{\gamma}{g^2}\leq \frac{1}{4}$. This ensures that $\frac{\nu}{1-\nu} = \frac{\gamma}{g^2}$ and we obtain
$$\gamma(\bar{H}(t)) \geq \frac{1}{3}\br{\frac{\frac{m^2\nu}{2(1-\nu)}}{1+\frac{m^2\nu}{2(1-\nu)}}}=\frac{1}{3}\cdot \frac{t^2\gamma}{128L^2g^2 +t^2\gamma} \geq  \frac{t^2\gamma}{400L^2g^2 + 3t^2\gamma}.$$
Substituting it in Equation \ref{coarsegrainedgaps}, we find
$$\frac{2\gamma}{\gamma(t)} \geq \frac{t^2\gamma}{400L^2g^2 + 3t^2\gamma} \implies \gamma(t)\leq \frac{10^3L^2g^2}{t^2} + 6\gamma.$$
This concludes the proof.

\section{Local verses global spectral gap on $D$ dimensional lattices}
\label{sec:Ddimgap}

Consider a $D$ dimensional regular lattice $\cL=[1: n_1]\times [1: n_2]\times \ldots [1: n_d]$ and let $$H_{\cL} = \sum_{\vi} P_{\vi}$$ be a frustration-free local hamiltonian, where the index $\vi$ enumerates the unit cells of the lattice and $P_{\vi}$ acts non-trivially only on the vertices of the $\vi$-th unit cell. Let $\gamma$ be the spectral gap of $H_{\cL}$. Since Theorem \ref{theo:knabe} also applies to periodic chains, the results below can similarly be extended to hamiltonians with periodic boundary conditions on the lattice. We study this model as an illustrative example, and highlight that the results below easily generalize for any local hamiltonian of constant locality on the lattice. 

In the above setting, we have $L,g \leq (3D)^D$. For a region $R\subseteq \cL$, let $\gamma(R)$ be the spectral gap of the hamiltonian
$$H_R = \sum_{{\vi}: P_{\vi}\in \text{supp}(R)} P_{\vi}.$$
 For integers $t_1, \ldots t_D$, we define $\gamma(t_1, \ldots t_D)$ as the minimum of $\gamma(R)$ over all hyper-rectangular regions $R$ of dimension $t_1\times t_2\times \ldots t_D$. Formally,
$$\gamma(t_1, \ldots t_D) = \min_{a_1, a_2, \ldots a_D:0\leq a_i\leq n_i-t_i}\gamma\bigg([a_1+1:a_1+t_1]\times [a_2+1:a_2+t_2]\times \ldots [a_D+1:a_D+t_D]\bigg).$$
We show the following theorem.
\begin{theorem}
\label{theo:rectknabe}
Suppose $2^64^{D}L < t_s < n_s/5$ for all $s\in [1:D]$ and $\gamma\leq \frac{g^2}{16^{D}}$. It holds that
$$\gamma(t_1, t_2, \ldots t_D) \leq 6^D\gamma +  200L^2g^26^{D}\cdot\frac{1}{\min_qt_q^2}$$
\end{theorem}
\begin{proof}
The proof will follow by inductive application of Theorem \ref{theo:knabe}. 
\begin{itemize}
\item {\bf Base case:} We view $H_{\cL}$ as a hamiltonian on a one dimensional chain of large qudits. This is achieved by combining the qudits $\{i\}\times [1:n_2]\ldots \times [1:n_D]$ into a single $i$th qudit of the chain. Defining $$H_i\defeq \sum_{{\vi} :P_{\vi}\in \text{supp}\br{\{i,i+1\}\times [1:n_2]\ldots \times [1:n_D]}}P_{\vi}$$ (c.f. Equation \ref{eq:Hidecompose}), we obtain the identity $H_{\cL}= \sum_{i=1}^{n_1-1} H_i$, which is the decomposition given in Equation \ref{eq:twoloc}. This allows us to conclude, from Theorem \ref{theo:knabe}, that 
\begin{equation}
\label{eq:gap1}
\gamma\br{t_1, n_2, \ldots n_D} \leq  \frac{10^3L^2g^2}{t_1^2} + 6\gamma.
\end{equation}
Since $$\gamma \leq \frac{1}{16^{D-1}}\cdot\frac{g^2}{16}, \quad \text{and}\quad \frac{10^\frac{3}{2}Lg}{t_1}= \frac{1}{4^D}\cdot\frac{g}{2}\cdot\frac{2\cdot 10^{\frac{3}{2}}L4^{D}}{t_1} \leq \frac{1}{4^{D}}\cdot\frac{g}{2},$$
Equation \ref{eq:gap1} additionally implies that 
\begin{equation}
\label{eq:gapcondition1}
\gamma\br{t_1, n_2, \ldots n_D} \leq  \frac{1}{16^D}\cdot\frac{g^2}{4}+\frac{6}{16^{D-1}}\cdot\frac{g^2}{16}  \leq \frac{g^2}{16^{D-1}} < \frac{g^2}{4},
\end{equation}
 maintaining the condition on spectral gap in Theorem \ref{theo:knabe}. 

\item {\bf Recursion:} Fix an $s\in [2:D]$. Assume
\begin{equation}
\label{eq:gapcondition2}
\gamma\br{t_1, t_2, \ldots  t_{s-1}, n_{s} \ldots n_D} \leq \frac{g^2}{16^{D-s+1}} < \frac{g^2}{4},
\end{equation}
which is true for $s=2$ via Equation \ref{eq:gapcondition1} and for $s>2$ via Equation \ref{eq:gapcondition2} in previous recursion. Let $$R=[a_1+1:a_1+t_1]\times \ldots [a_{s-1}+1:a_{s-1}+t_{s-1}]\times [1:n_s] \times \ldots [1:n_D]$$ be a hyper-rectangle that achieves the minimum in the definition of $$\gamma\br{t_1, t_2, \ldots  t_{s-1}, n_{s} \ldots n_D}.$$ Defining 
$$H'_i\defeq \sum_{{\vi} :P_{\vi}\in \text{supp}\br{[a_1+1:a_1+t_1]\times\ldots [a_{s-1}+1:a_{s-1}+t_{s-1}]\times\{i,i+1\}\times [1:n_{s+1}]\ldots \times [1:n_D]}}P_{\vi},$$ we have the decomposition
$$H_R = \sum_{i=1}^{n_s-1} H'_i,$$
which is the same as given in Equation \ref{eq:twoloc}. Since the values of $g,L$ remain unchanged for $H_R$,
we can apply Theorem \ref{theo:knabe} (c.f. Equation \ref{eq:gapcondition2}) and obtain the relation
\begin{equation}
\label{eq:gaprecurse} 
\gamma\br{t_1, t_2, \ldots  t_s, n_{s+1} \ldots n_D}\leq \frac{10^3L^2g^2}{t_s^2} + 6\gamma(R) = \frac{10^3L^2g^2}{t_s^2} + 6\gamma\br{t_1, t_2, \ldots  t_{s-1}, n_{s} \ldots n_D}.
\end{equation}
 Using Equations \ref{eq:gapcondition2} and \ref{eq:gaprecurse} we further have 
$$\gamma\br{t_1, t_2, \ldots  t_s, n_{s+1} \ldots n_D} \leq  \frac{1}{16^D}\cdot \frac{g^2}{4}+\frac{8g^2}{16^{D-s+1}} \leq \frac{g^2}{16^{D-s}}.$$
This ensures that Equation \ref{eq:gapcondition2} continues to be satisfied as we update $s\rightarrow s+1$.
\end{itemize}

Having obtained Equation \ref{eq:gaprecurse} for all $s\in [2:D]$ and Equation \ref{eq:gap1}, we combine them to arrive at the upper bound
$$\gamma(t_1, t_2, \ldots t_D) \leq 6^D\gamma+10^3L^2g^2\cdot\br{\sum_{q=1}^D \frac{6^{D-q}}{t_q^2}}\leq 6^D\gamma + 10^3L^2g^2\frac{6^{D}}{5}\cdot\frac{1}{\min_qt_q^2}.$$ This concludes the proof. 
\end{proof}

The dependence on $\min_q t_q^2$ cannot be improved; although the dependence on $D$ might not be optimal. To show this, we provide the following example adapted from \cite{GossetM16}. We consider the heisenberg ferromagnet, which is a one dimensional chain of qubits with frustration-free local hamiltonian defined by nearest-neighbour interaction $\frac{1}{2}\br{\ket{01}-\ket{10}}\br{\bra{01}-\bra{10}}$. The spectral gap of an open chain of length $n_1$ is $\frac{\pi^2}{2n_1^2}$.   We take $n_2\times n_3\times \ldots n_D$ independent copies of this system and arrange them on a $D$ dimensional lattice, with the chains running in the `first' dimension. That is, for each $i_2, \ldots i_D\in [1:n_2]\times \ldots [1:n_D]$, the set of qubits $\{(i, i_2, i_3,\ldots i_D)\}_{i=1}^{n_1}$ interact via the nearest-neighbour heisenberg ferromagnetic interaction. Consider all hyper-rectangles of dimension $t\times n_2 \times \ldots n_D$. Any such hyper-rectangle contains $n_2\times n_3\times \ldots n_D$ independent copies of the heisenberg ferromagnetic chain of length $t$, and hence the local spectral gap in this hyper-rectangle is the minimum local spectral gap of each copy, which is $\frac{\pi^2}{2t^2}$. Equivalently, $\gamma(t, n_2, \ldots n_D) = \frac{\pi^2}{2t^2}$. On the other hand, in the limit $n_1, n_2, \ldots n_D\rightarrow \infty$, we have $\gamma\rightarrow 0$. Since $t$ is the smallest of $\{t, n_2, \ldots n_D\}$,  this saturates the bound in Theorem \ref{theo:rectknabe} (up to the factors that depend on $D$).

The definition of $\gamma(t_1, t_2, \ldots t_D)$ takes a minimum over all hyper-rectangles of dimension $t_1\times t_2\times \ldots t_D$. To see that this is cannot be improved to an average of the spectral gap over all hyper-rectangles, consider the following hamiltonian for $D=1$: $$H= \sum_{i=1}^{k-1} P_{i,i+1} + \sum_{i=k+1}^nP'_i,$$ where $P_i=\frac{1}{2}\br{\ket{01}-\ket{10}}\br{\bra{01}-\bra{10}}$ and $P'_i= \ketbra{1}$. This is the same heisenberg ferromagnet on the first $k$ qubits and a trivial hamiltonian on the rest. For this hamiltonian, $\gamma = \bigo{\frac{1}{2k^2}}$. But the spectral gap, averaged over all hamiltonians on line segments of length $k$, is at least $1-\frac{k}{n}$. This is much larger than $\gamma + \frac{1}{k^2}$.

\section{Proof of Lemma \ref{convdetect}}
\label{append:convdetect}
\begin{proof}
Define two projectors
$$\Pi_1\defeq \prod_{i,j\in T_1}(\id-P_{ij}), \Pi_2\defeq \prod_{i,j\in T_2}(\id-P_{ij}).$$ 
Since $P_{ij}$ mutually commute for all $i,j\in T_{\alpha}$, we have
\begin{eqnarray*}
\Pi_1 \succeq \id - \sum_{i,j\in T_1} P_{ij}, \quad \Pi_2 \succeq \id - \sum_{i,j\in T_2} P_{ij}.
\end{eqnarray*}
Adding both sides, we find
\begin{equation}
\label{commutingsucc}
\Pi_1+\Pi_2 \succeq 2\id - \br{\sum_{i,j\in T_1} P_{ij} + \sum_{i,j\in T_2} P_{ij}} = 2\id - H. 
\end{equation}
Next, we apply Jordan's lemma \cite{jordan1875}, which states that $\Pi_1$ and $\Pi_2$ can be simultaneously block diagonalized in the following sense. There exist orthogonal projectors $\bar{\Pi}_{\beta}$ of dimension at most $2$, such that
$$\Pi_{\alpha}= \sum_{\beta} \bar{\Pi}_{\beta}\Pi_{\alpha}\bar{\Pi}_{\beta}, \quad \forall \alpha\in\{0,1\}.$$
Moreover, $\ketbra{v_{\alpha,\beta}}\defeq\bar{\Pi}_{\beta}\Pi_{\alpha}\bar{\Pi}_{\beta}$ is either a one dimensional normalized vector or a null vector. As a consequence, we have the identities
\begin{equation}
\label{eq:jordanid}
\Pi_2\Pi_1\Pi_2 = \sum_{\beta}|\braket{v_{1,\beta}}{v_{2,\beta}}|^2\ketbra{v_{2,\beta}}, \quad \Pi_1+\Pi_2=\sum_{\beta}\br{\ketbra{v_{1,\beta}}+\ketbra{v_{2,\beta}}}.
\end{equation}
We will show the following claim.
\begin{claim}
\label{clm:2dimop}
Let $0\leq \nu \leq \frac{3-\sqrt{5}}{2}$. It holds that 
$$\ketbra{v_{1,\beta}}+\ketbra{v_{2,\beta}} \preceq \nu|\braket{v_{1,\beta}}{v_{2,\beta}}|^2\ketbra{v_{2,\beta}} + (2-\nu)\bar{\Pi}_{\beta}.$$
\end{claim}
Before proving the claim, let us show how it implies the lemma. Setting $\nu=\frac{1}{3} < \frac{3-\sqrt{5}}{2}$ and substituting Claim \ref{clm:2dimop} in Equation \ref{eq:jordanid}, we find that
$$\Pi_1+\Pi_2 \preceq \frac{1}{3}\Pi_2\Pi_1\Pi_2 + (2-\frac{1}{3})\id = \frac{1}{3}DL^{\dagger}(H)DL(H)+ \frac{5}{3}\id.$$
Using this in Equation \ref{commutingsucc}, we obtain
$$2\id-H \preceq \frac{1}{3}DL^{\dagger}(H)DL(H)+ \frac{5}{3}\id \implies \frac{1}{3}\id-H \preceq \frac{1}{3}DL^{\dagger}(H)DL(H).$$
This proves the lemma after multiplying both sides by $\ket{\psi}$.
\begin{proof}[Proof of Claim \ref{clm:2dimop}]
Let $\ket{0}\defeq\ket{v_{2,\beta}}$ and $a\ket{0}+b\ket{1}=\ket{v_{1,\beta}}$, where $|a|^2+|b|^2=1$. The claimed inequality is equivalent, in matrix representation, to
$$\begin{pmatrix} 1+|a|^2 & ab^{*} \\ a^{*}b & |b|^2 \end{pmatrix} \preceq \nu|a|^2\begin{pmatrix} 1 & 0 \\ 0 & 0 \end{pmatrix} + (2-\nu)\begin{pmatrix} 1 & 0 \\ 0 & 1 \end{pmatrix}=\begin{pmatrix} 2-\nu|b|^2 & 0 \\ 0 & 2-\nu \end{pmatrix}.$$
This can be re-written as
$$0\preceq \begin{pmatrix} 1-|a|^2-\nu|b|^2 & -ab^{*} \\ -a^{*}b & 2-\nu-|b|^2 \end{pmatrix}=\begin{pmatrix} (1-\nu)|b|^2 & -ab^{*} \\ -a^{*}b & 1+|a|^2-\nu \end{pmatrix}.$$
Since the trace of the matrix on right hand side is positive for $\nu<1$, above inequality is satisfied if the determinant is non-negative. The determinant can be computed to be
$$(1+|a|^2-\nu)|b|^2(1-\nu)- |a|^2|b|^2=|b|^2\br{(1-\nu)^2+|a|^2(1-\nu)-|a|^2}=|b|^2\br{(1-\nu)^2-\nu|a|^2},$$
which is non-negative for all $\nu$ satisfying $(1-\nu)^2-\nu\geq 0$. This is satisfied if $\nu\leq  \frac{3-\sqrt{5}}{2}.$ This completes the proof.
\end{proof}
\end{proof}

\section{Proof of Lemma \ref{lightconelem}}
\label{append:lightcone}

\begin{proof}
The lower bound follows from Lemma \ref{convdetect}. The upper bound uses the following claim, adapted from \cite{GossetH15}. 
\begin{claim}[See Claim B.1, \cite{AAG19}]
\label{clm:chebysevDL}
Let $F$ be any polynomial of degree at most $\lceil\frac{t}{8L}\rceil$ such that $F(1)=1$. It holds that 
\begin{equation}
DL(t)= \br{\prod_{k} \br{\id - Q_{S_k}}}F\br{DL(H)^{\dagger}DL(H)}\br{\prod_{k} \br{\id - Q_{T_k}}}.
\end{equation} 
\end{claim}
Before outlining the proof of this claim, note that we can set $F=\step_{\frac{t}{8L}, \frac{\gamma}{g^2+\gamma}}$ to obtain   
\begin{eqnarray*}
&&\max_{\psi\in G_{\perp}}\|DL(t)\ket{\psi}\|^2\\
&&= \max_{\psi\in G_{\perp}}\|\br{\prod_{k} \br{\id - Q_{S_k}}}\step_{\frac{t}{8L}, \frac{\gamma}{g^2+\gamma}}\br{DL(H)^{\dagger}DL(H)}\br{\prod_{k} \br{\id - Q_{T_k}}}\ket{\psi}\|^2\\
&&\leq \max_{\psi\in G_{\perp}}\|\step_{\frac{t}{8L}, \frac{\gamma}{g^2+\gamma}}\br{DL(H)^{\dagger}DL(H)}\ket{\psi}\|^2.
\end{eqnarray*}
In the last inequality, we used the following :
$$\br{\prod_{k} \br{\id - Q_{T_k}}}\ket{\psi}\in G_{\perp}, \quad \|\br{\prod_{k} \br{\id - Q_{T_k}}}\ket{\psi}\|\leq 1.$$ From Lemma \ref{detectlem}, the second largest eigenvalue of $DL(H)^{\dagger}DL(H)$ is at most $\frac{1}{1+\frac{\gamma}{g^2}}= 1- \frac{\gamma}{g^2+\gamma}$. This concludes the proof of Lemma \ref{lightconelem}.

\begin{proof}[Proof outline of Claim \ref{clm:chebysevDL}]
Following \cite{AAG19}[Claim B.1], we consider the `layer operators'
$$DL_{\alpha}\defeq \prod_{i,j: P_{ij}\in T_{\alpha}}\br{\id-P_{ij}}.$$
Observe that $DL(H) = DL_1DL_2\ldots DL_L$ and hence $$DL(H)^{\dagger}DL(H) = DL_L\ldots DL_2DL_1DL_2\ldots DL_L.$$ This implies that the operator $$\br{DL(H)^{\dagger}DL(H)}^{q}= \br{DL_L\ldots DL_2DL_1DL_2\ldots DL_{L-1}}^{q-1}DL_L\ldots DL_2DL_1DL_2\ldots DL_L$$ is a product of $(2L-2)\cdot(q-1)+2L-1 = q\cdot(2L-2) +1 $ operators $DL_{\alpha}$. Suppose we have $$q\leq \lceil\frac{t}{8L}\rceil \implies q\cdot(2L-2) +1 < \lfloor t/4\rfloor \quad (\text{using } t\geq 8L^2).$$ Since the overlap between an $S$ set and the adjacent $T$ set is at least $\lfloor t/4\rfloor$ (Equation \ref{eq:overlaplb}), all the operators can be `absorbed' in either $\br{\prod_{k} \br{\id - Q_{S_k}}}$ or $\br{\prod_{k}\br{\id - Q_{T_k}}}$. This ensures that 
$$\br{\prod_{k} \br{\id - Q_{S_k}}}\br{DL(H)^{\dagger}DL(H)}^{q}\br{\prod_{k} \br{\id - Q_{T_k}}}= \br{\prod_{k}\br{\id - Q_{S_k}}}\br{\prod_{k} \br{\id - Q_{T_k}}}.$$  This proves the claim if we take the linear combination of above equation according to the polynomial $F$. 
\end{proof}
\end{proof}

\end{document}